\documentclass[journal]{IEEEtran}
\usepackage{csquotes}
\usepackage{amsmath}
\usepackage{amsmath, amsthm, amssymb}
\usepackage{algorithm}
\usepackage{algorithmic}
\usepackage{csquotes}
\usepackage{amsmath}

\usepackage{authblk}
\usepackage{hyperref}
\usepackage{enumitem}
\usepackage{float}
\usepackage{cite}
\usepackage[table,xcdraw]{xcolor}
\usepackage{multirow}

\usepackage[bottom]{footmisc}
\usepackage{mathtools}

\ifCLASSINFOpdf
  \usepackage[pdftex]{graphicx}
\else
\fi
\usepackage{amsmath}
\usepackage{amsfonts}
\usepackage{amssymb}
\usepackage{color}

\usepackage{algorithmic}
\usepackage{algorithm}

\usepackage{array}

\ifCLASSOPTIONcompsoc
 \usepackage[caption=false,font=normalsize,labelfont=sf,textfont=sf]{subfig}
\else
\usepackage[caption=false,font=footnotesize]{subfig}
\fi
\usepackage{fixltx2e}

\usepackage{stfloats}
\usepackage{url}

\usepackage{amsmath, amsthm, amssymb}
\newtheorem{corollary}{Corollary}
\newtheorem{proposition}{Proposition}
\newtheorem{remark}{Remark}

\usepackage{booktabs}
\usepackage{varwidth}
 \usepackage[table,xcdraw]{xcolor}
\begin{document}

%
\title{Simultaneous Transmitting and Reflecting  Reconfigurable Intelligent Surfaces-Empowered NOMA Networks}
%
%
%

\author{Mahmoud~Aldababsa, Aymen~Khaleel,
	and~Ertugrul~Basar,~\IEEEmembership{Fellow~Member,~IEEE}\vspace{-0.2cm}
	\thanks{		
		Mahmoud Aldababsa is with the Department of Electrical and Electronics Engineering, Nisantasi University, 34481742, Istanbul, Turkey. (e-mail: mahmoud.aldababsa@nisantasi.edu.tr).}
	\thanks{Aymen Khaleel and Ertugrul Basar are with the Communications Research and Innovation Laboratory (CoreLab), Department of Electrical and Electronics Engineering, Koc University, Sariyer 34450, Istanbul, Turkey. (e-mail: akhaleel@ku.edu.tr, ebasar@ku.edu.tr).}
	\thanks{}
	\thanks{}}

\maketitle
\begin{abstract}
In this paper, we propose a novel simultaneous transmitting and reflecting reconfigurable intelligent surface (STAR-RIS) assisted non-orthogonal multiple access (NOMA) system. Unlike most of the STAR-RIS-assisted NOMA works, we target scalable phase shift design for our proposed system that requires a reduced channel estimation overhead. Within this perspective, we propose novel algorithms to partition the STAR-RIS surface among the available users. These algorithms aim to determine the proper number of transmitting/reflecting elements needs to be assigned to each user in order to maximize the system sum-rate while guaranteeing the quality-of-service requirements for individual users. For the proposed system, we derive closed-form analytical expressions for the outage probability (OP) and its corresponding asymptotic behavior under different user deployments. Finally, Monte Carlo simulations are performed in order to verify the correctness of the theoretical analysis. It is shown that the proposed system outperforms the classical NOMA and orthogonal multiple access (OMA) systems in terms of OP and sum-rate, under spatial correlation and phase errors.
\end{abstract}
%
\IEEEpeerreviewmaketitle
\section{Introduction}
\IEEEPARstart{R}{econfigurable} intelligent surface (RIS)-empowered communication has been considered as a promising candidate to enhance the performance of future wireless networks in terms of energy efficiency and coverage \cite{liu2021}-\hspace{-0.01cm}\cite{sb2021}. Basically, an RIS consists of massive low-cost reconfigurable passive elements by which it can reconfigure the propagation of incident wireless signals by adjusting the amplitude and phase shift of each individual element \cite{wu2020}. RISs do not require radio frequency chains, rather, they process the impinging signal over-the-air. This remarkably reduces the energy consumption and hardware costs, and hence, makes RISs economical and environmentally friendly compared to multi-antenna and relaying systems \cite{eb2021}. Despite the aforementioned advantages, RISs act as reflective metasurfaces, which means that RISs can only serve the users located on their same side. In other words, RISs offer only half-space coverage, which limits the flexibility of deploying them. To solve this problem, recently, the novel concept of simultaneous transmitting and reflecting RISs (STAR-RISs) has been proposed \cite{yua2021}. Compared to conventional RISs, STAR-RISs have elements that support both electric polarization and magnetization currents yielding simultaneous control of the transmitted and reflected signals. In this way, STAR-RISs can offer full-space coverage. Furthermore, in order to enhance this full-space coverage, three main operating protocols are proposed, namely the energy splitting (ES), mode switching (MS), and time switching (TS). Specifically, for ES, all elements of the STAR-RIS are assumed to operate in transmission and reflection mode simultaneously. For MS, each element can operate in full transmission or reflection mode, while in TS, all elements periodically switch between the transmission and reflection modes in orthogonal time slots.

Non-orthogonal multiple access (NOMA) has also received significant attention due to its ability to realize high spectral efficiency (SE), massive connectivity, and low latency \cite{yau2021}. NOMA is fundamentally different than conventional OMA schemes which provide orthogonal access to the users either in time, frequency, code, or space. The key idea in NOMA is to allocate non-orthogonal resources to serve multiple users, yielding a higher SE while allowing some degree of interference at receivers \cite{mah2018}. In power-domain NOMA (PD-NOMA), multiple users' signals are superposed with different power levels that are in reverse order to their channels' gains. At the receiver side, successive interference cancellation (SIC) is applied by each user to recover its signal, providing a good trade-off between system throughput and user fairness \cite{noma2017}.
\subsection{Related Works}
\subsubsection{Studies on STAR-RIS}
Since the concept of STAR-RIS is  relatively new, only a few works in the literature have investigated it so far \cite{star22021}-\hspace{-0.01cm}\cite{star112021}. Particularly, in \cite{star22021}, a general hardware model for STAR-RISs was introduced. The channel models for the near- and far-field regions of STAR-RISs were also investigated and the analytical expressions for the channel gains of users were derived in closed-form. Next, the active and passive beamforming optimization problem was considered in \cite{star42021} for the STAR-RIS-assisted downlink networks in both unicast and multicast transmission cases. Especially, for three practical operating protocols ES, MS, and TS, the active beamforming at the base station (BS) and the passive transmission/reflection beamforming at the STAR-RIS were jointly optimized in order to minimize the power consumption of the BS and satisfy the quality-of-service (QoS) requirements of the users. Then, the weighted sum-rate maximization problem of the STAR-RIS-assisted multiple-input multiple-output (MIMO) networks was studied in \cite{star72021}. Here, the authors proposed an alternative block coordinate descent algorithm to design the precoding matrices and the transmitting/reflecting coefficients. In \cite{star112021}, practical hardware implementations and their accurate physical models were investigated for STAR-intelligent omni surfaces (STAR-IOSs).

\subsubsection{Studies on STAR-RIS-assisted NOMA}
Due to their remarkable advantages, the integration of RISs with NOMA systems is indispensable in future wireless communications to fulfill the stringent demands on data rate and connectivity \cite{yau2021}. Motivated by this, recent research efforts have been devoted to STAR-RIS-assisted NOMA systems \cite{star32021}-\hspace{-0.01cm}\cite{starnoma++++++}. Particularly, the authors in \cite{star32021} maximized the coverage range of STAR-RIS aided two-user communication networks for both NOMA and OMA by jointly optimizing the resource allocation at the BS and the transmission/reflection coefficients at the STAR-RISs. Next, the authors in \cite{star52021} exploited a STAR-RIS in NOMA enhanced coordinated multi-point transmission networks in order to eliminate and boost the inter-cell interferences and desired signals, respectively. Then, in \cite{star62021}, by jointly optimizing the decoding order, power allocation coefficients, active beamforming, and transmission/reflection beamforming, the achievable sum-rate was maximized for the STAR-RIS-assisted NOMA system. The authors in \cite{star82021} used a STAR-RIS to enable a heterogeneous network that integrated uplink NOMA and over-the-air federated learning into a unified framework to address the scarcity of system bandwidth. In addition, they minimized the optimality gap while guaranteeing QoS requirements by jointly optimizing the transmit power at users and the configuration mode at the STAR-RIS. For the sake of deriving approximated mathematical channel models for STAR-IOS/RISs, the authors exploited the central limit, curve fitting and M-fold convolution models in \cite{star92021} while central limit theorem (CLT) and channel power gain models are used in \cite{star102021}. In \cite{Correlated-TR} and \cite{Coupled-PS}, the authors investigated the correlation effects associated with phase adjustment of the transmitting and reflecting elements, as follows. In \cite{Correlated-TR}, the authors proposed three phase-shift configuration strategies to achieve full diversity and enhance communication performance. In \cite{Coupled-PS}, the authors considered an optimization problem to minimize the power consumption while satisfying the users' rate requirements. The performance of STAR-RIS-assisted NOMA networks under Rician fading channels is investigated in \cite{STAR-Performance}, in terms of OP and ergodic rate. By using deep reinforcement learning, the authors in \cite{E-Efficiency} considered an optimization problem to maximize the energy efficiency. In \cite{Secrecy}, the authors investigated the secrecy design of STAR-RIS-assisted NOMA networks under two scenarios where full or statistical channel state information (CSI) of the eavesdropper is available at the legitimate nodes. For multi-carrier communications networks, the authors in \cite{R-Allocation} investigated the resource allocation problem in STAR-RIS-assisted OMA and NOMA networks by considering the sum-rate maximization problem in both scenarios. Finally, in \cite{Uplink}, the authors considered STAR-RIS-assisted NOMA networks where the total power consumption is minimized by jointly optimizing the users' transmit power, receive-beamforming vectors at the STAR-RIS, and time slots. The authors in \cite{star-ber} derived the bit error rate in a closed-form for a STAR-RIS-assisted NOMA network under perfect and imperfect SIC scenarios. The authors in \cite{starnoma+} minimized the total transmit power for STAR-RIS-empowered uplink NOMA systems. In \cite{starnoma++} and \cite{starnoma++A}, the authors investigated the ergodic rate for STAR-RIS-NOMA networks. The authors in \cite{starnoma+++} obtained analytical expressions of the effective capacity for the network with a pair of NOMA users on the different sides of the STAR-RIS.  
 
\subsubsection{Studies on partitioning algorithms}
In \cite{part1}, the authors introduced a general framework based on RIS partitioning for the channel estimation for large intelligent metasurface assisted massive MIMO networks. In \cite{part3}, the author brought the concept of RIS-assisted communications to the realm of IM by partitioning the RIS surfaces and embedding bits in the indices of these partitions. In \cite{part4}, the whole RIS was virtually partitioned into two halves to create signals with only in-phase and quadrature components, respectively, and each half forms a beam to a receive antenna whose index carries the bit information. In \cite{RISpartitioning}, the authors formulated an RIS partitioning optimization problem to slice the RIS elements between the users such that the user fairness is maximized.
\vspace{-0.4cm}
\subsection{Motivation and Contributions} 
Considering the phase shift design in the previously mentioned STAR-RIS-assisted NOMA systems, we note the following. Since the phase shifts are jointly designed for all users, the BS-RIS and RIS-user CSI of all users needs to be available at the BS side. Depending on the considered channel estimation method, $KN$ pilot signals might be required to obtain the RIS-user CSI of all users \cite{ch-est}, where $N$ and $K$ are the numbers of STAR-RIS elements and users, respectively, considering a single-input single-output (SISO) multi-user system. Consequently, this approach of phase shift design incurs a huge channel estimation overhead that adversely affects the overall system performance in terms of the data rate. Furthermore, considering the joint optimization of the phase shift design, power allocation at the BS, SIC decoding order, and QoS constraints for all users, the solution for this kind of optimization problem is non-scalable when any of these decision variables changes. Also, due to the heavy computations associated with solving such optimization problems, it is practically challenging to update the phase shift design in real-time.

Against this background and motivated by the work in \cite{RISpartitioning}, we propose a reduced channel estimation overhead and scalable phase shift design solution for a STAR-RIS-assisted NOMA system by partitioning the STAR-RIS among the available users to maximize the overall sum-rate while guaranteeing their individual QoS. In the proposed system, the BS communicates multiple users with the assistance of a STAR-RIS. The users are assumed to be randomly deployed on the transmission and reflection sides of the STAR-RIS. According to the STAR-RIS MS protocol\footnote{
The MS mode is more practical due to its ease of implementation compared to the other modes. In particular, the ES mode has a larger number of variables that need to be jointly optimized, while the TS mode has a high hardware implementation complexity \cite{yua2021}.}, each element of the STAR-RIS can operate in transmission or reflection mode. Therefore, in order to serve the users on both sides of the STAR-RIS, we extend the partitioning algorithms proposed in \cite{RISpartitioning} to the STAR-RIS case to partition the surface among users over two stages. First, the RIS surface is partitioned into two parts, transmission and reflection parts, which contain all the elements that operate in transmission and reflection modes, respectively. Second, each STAR-RIS part is partitioned into subsurfaces, where each subsurface is allocated to serve a specific user located on the side of that part. Note that, unlike our proposed system, \cite{RISpartitioning} suffers from strong interference from the direct link due to its transmission mechanism. Furthermore, it is limited to phase shift keying (PSK) modulation. Finally, it requires a high synchronization overhead between the transmission at the BS and the remodulation process at the STAR-RIS side. The main contributions of the paper can be summarized as follows:
\begin{itemize}
	\item We introduce a scalable phase shift design scheme for a STAR-RIS-assisted NOMA system with reduced channel estimation overhead by efficiently partitioning the STAR-RIS surface among users.
	\item We propose two novel algorithms to partition the STAR-RIS surface among users over two stages to maximize the system sum-rate while fulfilling the different QoS requirements for different users. 
	\item We derive exact and asymptotic expressions of the OP for the STAR-RIS aided NOMA network under different user deployments. 
	\item Using comprehensive computer simulations, we verify our theoretical results and show the superiority of the proposed STAR-RIS-assisted NOMA system over the classical NOMA and OMA systems in terms of OP and sum-rate, under spatial correlation and phase errors.
\end{itemize} 
\vspace{-0.3cm}
\subsection{Paper Organization and Notations }
The organization of the paper is given as follows: In Section \ref{sec:2}, the system model of the STAR-RIS-assisted NOMA network is described. The OP analysis is conducted in Section \ref{sec:3}. In Section \ref{sec:4}, the STAR-RIS partitioning algorithm is presented. In Section \ref{sec:5}, the numerical results are demonstrated and in Section \ref{sec:6}, the paper is concluded\footnote{{\it Notation}: Matrices and vectors are denoted by an upper and lower
case boldface letters, respectively. $\mathbf{X}\in \mathbb{C}^{m\times n}$ denotes a matrix $\mathbf{X}$ with $m\times n$ size. $\lfloor\mathbf{x}\rfloor_n$ denotes the $n$th entry of a vector $ \mathbf{x} $. $|\cdot|$, $\text{E}[\cdot]$, $ \lceil\cdot\rceil$ and $(\cdot)^{H}$ represent the absolute value, expectation operator, ceiling operator and Hermitian transpose, respectively. $ P_{r}(\cdot) $ symbolizes probability. $F_{X}\left(x\right)$ denotes cumulative distribution function (CDF) of a random variable $ X $. $ \mathcal{CN}(\mu,\sigma^{2})$ stands for the complex Gaussian distribution with mean $\mu$ and variance $\sigma^{2}$. $Q_{m}\left(\cdot,\cdot\right)$ refers to the $m$th order generalized Marcum Q-function.}. 
\vspace{-0.1cm}
\section{STAR-RIS NOMA Network: System Model}
\label{sec:2}
\begin{figure}[t]
	\includegraphics[width=45mm]{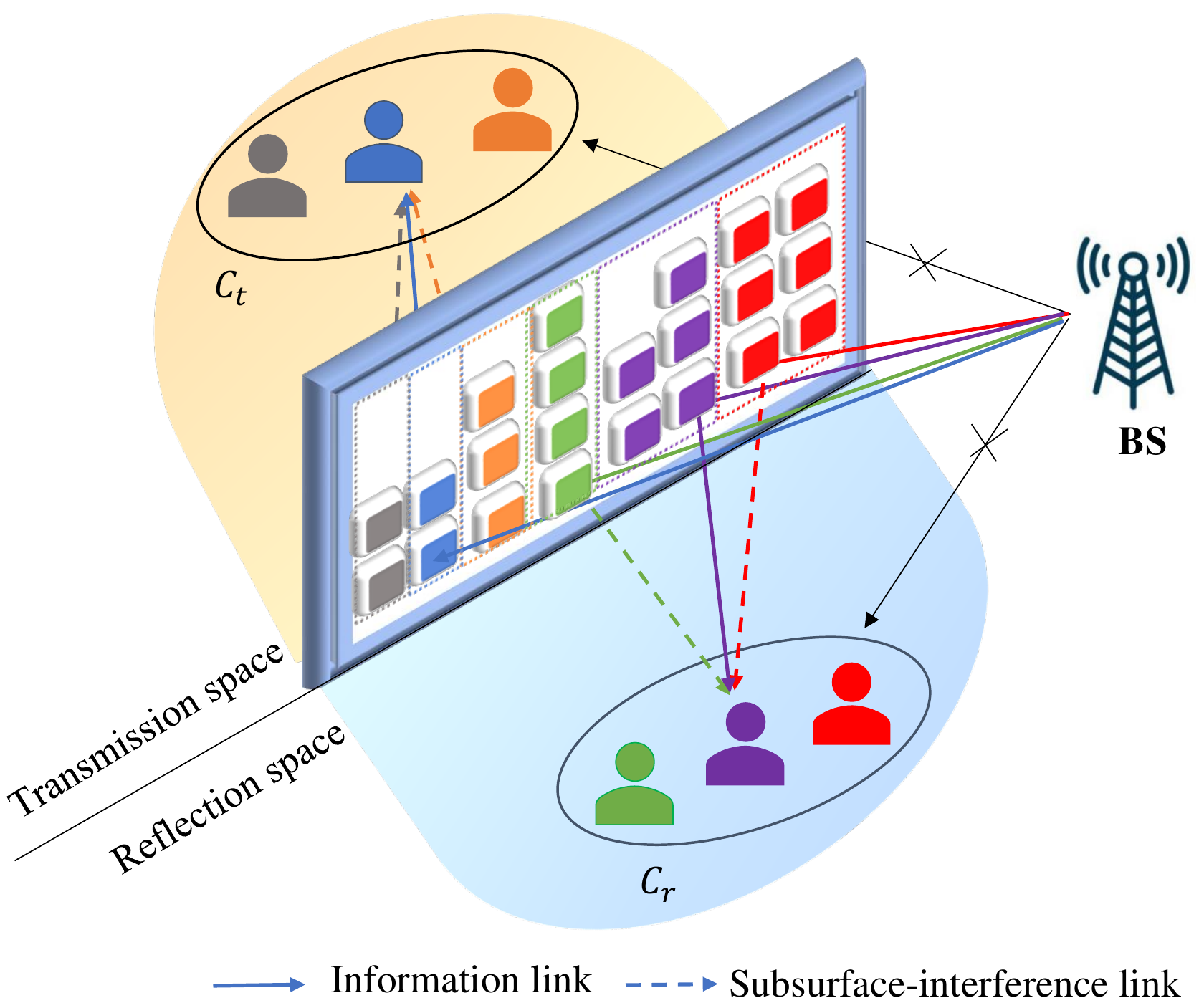}
	\centering
	\caption{A STAR-RIS aided multi-user downlink communication system model.}	
	\label{System model}
	\vspace{-0.7cm}
\end{figure}
Consider a downlink STAR-RIS-assisted multi-user wireless network where a single-antenna BS and an RIS with $N$ elements serve $K$ single-antenna users by using PD-NOMA, as shown in Fig. \ref{System model}. The users are grouped into two main clusters, $\text{C}_t$ and $\text{C}_r$, which contain $K_{t}$ users located in the transmission region and $K_{r}$ users located in the reflection region of the RIS, respectively. In order to reduce the channel estimation overhead associated with the STAR-RIS phase shift design problem and to make the solution for this problem scalable, we partition the STAR-RIS surface among users as follows. To serve the users in both clusters, the RIS is partitioned into two main parts, where the first and second parts contain $N_t$ and $N_r$ elements to serve the users in $\text{C}_t$ and $\text{C}_r$, respectively. The elements in the first and second parts are operated in the transmission ($t$) and reflection ($r$) modes, respectively, and the whole RIS is assumed to be deployed in the far-field of the BS. Furthermore, the transmission and reflection parts of the RIS are partitioned into subsurfaces, where each subsurface is allocated to serve a specific user in the transmission or reflection region. The BS-$U_k$ direct link is assumed to be blocked by obstacles, where $U_k$ stands for the user $k$, $k\in \mathbb{K}_{\chi}$, note that $\chi\in\{t,r\}$, $\mathbb{K}_t=\left\{1, ..., K_{t}\right\}$ and $\mathbb{K}_r=\left\{K_{t} + 1, ..., K_{t} + K_{r}\right\}$. The BS-RIS link is assumed to have a Rayleigh fading channel model, where $\mathbf{h}^i\in \mathbb{C}^{N^{i}_{\chi}\times1}$ denotes the BS-$i$th subsurface channel vector, $ N^{i}_{\chi} $ refers to the number of elements that belong to the $i$th subsurface in the transmission or reflection part of the RIS, and $\lfloor\mathbf{h}^i\rfloor_n=h^{i,n}$. Here, $h^{i,n}=\sqrt{L_{BS}}\zeta^{i,n}_{\chi}e^{-j\phi^{i,n}_{\chi}}$, where $L_{BS}$, $\zeta^{i,n}_{\chi}$, and $\phi^{i,n}_{\chi}$ denote the path gain, channel amplitude, and channel phase, respectively. Note that $ \zeta^{i,n}_{\chi} $ is independently Rayleigh distributed random variable (RV) with a mean $\mathrm{E}\left[\zeta^{i,n}_{\chi}\right]=\frac{\sqrt{\pi}}{2}$, a variance $\mathrm{VAR}\left[\zeta^{i,n}_{\chi}\right]=1-\frac{\pi}{4}$, and $\mathrm{E}\left[\left(\zeta^{i,n}_{\chi}\right)^{2}\right]= 1$. Hence, $h^{i,n}$ can be modelled as complex Gaussian RV with zero mean and a variance $L_{BS}$, i.e., $h^{i,n}\sim \mathcal{CN}(0,L_{BS})$. Likewise, the RIS-$\text{U}_{k}$ link is assumed to have Rayleigh fading channel model, where the channel vector between the $i$th subsurface in the transmission or reflection part of the RIS and $U_{k}$ is denoted by $\mathbf{g}^{i}_{\chi,k}$, where $\lfloor\mathbf{g}^{i}_{\chi,k}\rfloor_n=g^{i,n}_{\chi,k}$ and $\chi\in\{t,r\}$. Here, $g^{i,n}_{\chi,k} = \sqrt{L_{SU,\chi,k}}\eta^{i,n}_{\chi,k}e^{-j\varphi^{i,n}_{\chi,k}}$, where $L_{SU,\chi,k}$, $\eta^{i,n}_{\chi,k}$, and $\varphi^{i,n}_{\chi,k}$ denote the path gain, channel amplitude, and channel phase, respectively. Note that $ \eta^{i,n}_{\chi,k} $ is independently Rayleigh distributed RV with a mean $\mathrm{E}\left[\eta^{i,n}_{\chi,k}\right]=\frac{\sqrt{\pi}}{2}$, a variance $\mathrm{VAR}\left[\eta^{i,n}_{\chi,k}\right]=1-\frac{\pi}{4}$, and $\mathrm{E}\left[\left(\eta^{i,n}_{\chi,k}\right)^{2}\right]= 1$. Hence, $g^{i,n}_{\chi,k}$ can be modelled as complex Gaussian RV with zero mean and a variance $L_{SU,\chi,k}$, i.e., $g^{i,n}_{\chi,k}\sim \mathcal{CN}(0,L_{SU,\chi,k})$. All the considered channel coefficients are assumed to be perfectly available at the BS side. Although it is challenging to obtain CSI in RIS-assisted systems, yet, we assume the use of one of the methods proposed in the literature \cite{ch-est,ch-est1,ch-est2}. The transmission/reflection coefficients for the $i$th subsurface in the transmission or reflection part of the RIS are denoted by the entries of the diagonal matrix $\mathbf{\Theta}^{i}_{\chi}\in \mathbb{C}^{N^{i}_{\chi}\times N^{i}_{\chi}}$, for the $n$th element we have $\lfloor\mathbf{\Theta}^{i}_{\chi}\rfloor_n=\beta^{i,n}_{\chi}e^{j\theta^{i,n}_{\chi}}$, where $\theta^{i,n}_{\chi}\in[0,2\pi)$ and $\beta^{i,n}_{\chi}=1$, under the full transmission/reflection assumption.
\vspace{-0.1cm}

\indent Considering PD-NOMA, the BS transmits the signal $x$, which superimposes the $K_t$ and $K_r$ users' symbols, as follows:
\vspace{-0.1cm}
\setlength\abovedisplayskip{3pt}
\setlength\belowdisplayskip{3pt}
\begin{align} 
	\label{superimposedsignal}
	x &=\sum_{k = 1}^{K}\sqrt{a_{k}P}x_{k} \nonumber\\
	\
	&= \underbrace{\sum_{k = 1 }^{K_{t}}\sqrt{a_{k}P}x_{k}}_{\text{$K_t$ users' symbols}} + \underbrace{\sum_{k = K_{t} + 1 }^{K_{t} + K_{r}}\sqrt{a_{k}P}x_{k}}_{\text{$K_r$ users' symbols}},
\end{align}
where $P$ is the BS transmit power, $x_k$ and $a_k$ are the user $k$'s symbol and power allocation factor, respectively, ${\text{E}}\left[|x_{k}|^{2}\right]=1$, $\sum_{k = 1}^{K}a_{k} = 1$. The users are ordered according to their channel gains as $\{U_1, ...,  U_{K_t}, U_{K_t+1}, ...,  U_{K_t+K_r}\}$, where $U_1$ and $U_{K_t+K_r}$ are the users with the weakest and strongest channel gains, respectively. Thus, the power allocation coefficients are ordered as  $a_{1}\ge...\ge a_{k}\ge...\ge a_{K}$. By using the SIC process, each user first decodes all strongest signals' users and subtracts them from the received signal, and finally, it decodes its own message. Specifically, the $U_{k}$ cancels the interfered
signals from $U_{1}, ..., U_{k-1}$ and treats $U_{k+1}, ..., U_{K}$ as interference. Accordingly, considering successful SIC process at the $U_{k}$, the received signal can be expressed as
\begin{align} 
	\label{receivedsignal}
	y_{k} &= \underbrace{\left(\mathbf{g}^{k}_{\chi,k}\right)^{T}\hspace{-0.1cm}\mathbf{\Theta}^{k}_{\chi}\mathbf{h}^{k}_{\chi}\sqrt{a_{k}P}x_{k}}_{\text{Desired signal term}}+\underbrace{\left(\mathbf{g}^{k}_{\chi,k}\right)^{T}\hspace{-0.1cm}\mathbf{\Theta}^{k}_{\chi}\mathbf{h}^{k}_{\chi}\sum_{l= k+1}^{K}\sqrt{a_{l}P}x_{l}}_{\text{User-interference term}}\nonumber\\
\
& + \underbrace{\sum_{\underset{i\ne k}{i \in\mathbb{K}_{\chi}}}^{}\left(\mathbf{g}^{i}_{\chi,k}\right)^{T}\mathbf{\Theta}^{i}_{\chi}\mathbf{h}^{i}_{\chi}x}_{\text{Subsurface-interference term}} + n_{k},
\end{align}
where the first three terms in \eqref{receivedsignal} denote the desired signal of the $U_{k}$, the interference of other users, and the subsurfaces mutual interference within the transmission or reflection part of the RIS, respectively, while ${n}_{k}$ denotes the complex additive white Gaussian noise (AWGN) sample with zero mean and variance $\sigma^{2}_{k}$ at the $U_{k}$, ${n}_{k}\sim \mathcal{CN}(0,\sigma^{2}_{k})$.
\vspace{-0.0cm}

By letting $r^{k}_{\chi,k} = \left(\mathbf{g}^{k}_{\chi,k}\right)^{T}\mathbf{\Theta}^{k}_{\chi}\mathbf{h}^{k}_{\chi}$, $r^{i}_{\chi,k} = \left(\mathbf{g}^{i}_{\chi,k}\right)^{T}\mathbf{\Theta}^{i}_{\chi}\mathbf{h}^{i}_{\chi}$, $\sigma^2_{k}=\sigma^2$, and defining $\rho = \frac{P}{\sigma^{2}}$ is the transmit SNR, then the signal-to-interference-plus-noise ratio (SINR) for user $k$ to detect user $j$'s signal  $\gamma^{j}_{{k}}$, $\left\{j\le k, j\ne K\right\}$, can be obtained from \eqref{receivedsignal} as
\vspace{-0.0cm}
\begin{align}
	\label{SINR}
	\gamma^{j}_{{k}} &= \frac{\rho\Big|r^{k}_{\chi,k}\Big|^{2}a_{j}}{\rho\Big|r^{k}_{\chi,k}\Big|^{2}\sum_{l= j +1 }^{K}a_{l} + \rho\Big|\sum_{\underset{i\ne k}{i \in\mathbb{K}_{\chi}}}^{}r^{i}_{\chi,k}\Big|^{2} + 1}.
\end{align}
The phase shift design is considered for each user independent from the others. In this way, each subsurface $i$ tries to maximize $\gamma^{j}_{{k}}$ by adjusting the phase shifts of its elements to remove the overall BS-STAR-RIS-$\text{U}_k$ channel phases,
as follows:
\vspace{-0.3cm}
\begin{align}
\theta^{k,n}_{\chi}=\phi^{k,n}_{\chi} + \varphi^{k,n}_{\chi,k}.\label{eq:phs-adjust}
\end{align}
Accordingly, changing the CSI of any user and/or the number of users does not affect the phase adjustment of the other users. On the other side, in the joint phase shift design approach adopted in most of the STAR-RIS NOMA works in the literature (for example, \cite{star82021}, \cite{star92021}), the change in the CSI of even a single individual element requires the redesign of the phase shifts of all the remaining elements in the whole STAR-RIS surface. This gives the proposed partitioning scheme scalability in terms of the phase shift design at the expense of some performance loss compared to the joint phase shift design approach.

\vspace{-0.1cm}
By considering \eqref{eq:phs-adjust}$, \big|r^{k}_{\chi,k}\big|^{2} $ and $ \big|\sum_{\underset{i\ne k}{i \in\mathbb{K}_{\chi}}}^{}r^{i}_{\chi,k}\big|^{2} $ can be re-expressed, respectively, as
\vspace{-0.3cm}
\begin{align}
	\label{rk}
	\Big|r^{k}_{\chi,k}\Big|^{2}&=\bigg|\sqrt{L_{k}} \sum_{n = 1}^{N^{k}_{\chi}} \zeta^{k,n}_{\chi}\eta^{k,n}_{\chi,k}\bigg|^{2},\\ \Big|\sum_{\underset{i\ne k}{i \in\mathbb{K}_{\chi}}}^{}r^{i}_{\chi,k}\Big|^{2}&= \bigg|\sqrt{L_{k}}\sum_{\underset{i\ne k}{i \in\mathbb{K}_{\chi}}}^{} \sum_{n = 1}^{N^{i}_{\chi}} \zeta^{i,n}_{\chi}\eta^{i,n}_{\chi,k}e^{j\Phi^{i,n}_{\chi,k}} \bigg|^{2},
\end{align}
where $\Phi^{i,n}_{\chi,k} = \theta^{i,n}_{\chi} -\phi^{i,n}_{\chi}- \varphi^{i,n}_{\chi,k}$, $L_{k}={L_{BS}L_{SU,\chi,k}}$ is the overall path gain of the BS-STAR-RIS-$U_{k}$ link, which can be expressed as $L_{k}=  \rho^{2}_0/(d_{BS}^{\alpha_{BS}}d_{SU,\chi,k}^{\alpha_{SU}})$. Here, $d_{BS}$ is the distance between the BS and STAR-RIS and $d_{SU,\chi,k}$ is the distance between the STAR-RIS and the $U_{k}$ in the transmission or reflection group. Additionally, $\rho_0$, $\alpha_{BS}$, and $\alpha_{SU}$ are the path gain at a reference distance of 1 meter, the path loss exponents associated with the BS-STAR-RIS and the STAR-RIS-$U_k$ links, respectively.

By considering SIC, the SINR expression for user $k$, $k>1$, after successfully detecting and subtracting the $k-1$ users' signals before it, can be obtained as
\begin{align}
	\label{SINR-k}
	\gamma_{{k}} &= \frac{\rho\Big|r^{k}_{\chi,k}\Big|^{2}a_{k}}{\rho\Big|r^{k}_{\chi,k}\Big|^{2}\sum_{l= k +1 }^{K}a_{l} + \rho\Big|\sum_{\underset{i\ne k}{i \in\mathbb{K}_{\chi}}}^{}r^{i}_{\chi,k}\Big|^{2} + 1}.
\end{align}
For the $U_{k}$, we have
\vspace{-0.2cm}
\begin{align}
	\label{SINR-K}
	\gamma_{{K}} &= \frac{\rho\Big|r^{K}_{\chi,K}\Big|^{2}a_{K}}{ \rho\Big|\sum_{\underset{i\ne K}{i \in\mathbb{K}_{\chi}}}^{}r^{i}_{\chi,K}\Big|^{2} + 1}.
\end{align}
\vspace{-0.1cm}
Considering the special case of $\left(K_{t}, K_{r}\right) = \left(1, 1\right)$, where one user is located in the transmission region and the other user is located in the reflection region, the SINR of the $U_{1}$ to detect its own signal is given by
\vspace{-0.3cm}
\begin{align}
	\label{SINR-11}
	\gamma_{{1}} &= \frac{\rho\Big|r^{1}_{t,1}\Big|^{2}a_{1}}{\rho\Big|r^{1}_{t,1}\Big|^{2}a_{2} +  1},
\end{align}
and the SINR of the $U_{2}$ to detect the $U_{1}$'s signal is given by
\begin{align}
	\label{SINR-1}
	\gamma^{1}_{{2}} &= \frac{\rho\Big|r^{2}_{r,2}\Big|^{2}a_{1}}{\rho\Big|r^{2}_{r,2}\Big|^{2}a_{2} +  1},
\end{align}
finally, the SNR of the $U_{2}$ to detect its own signal is given by
\begin{align}
	\label{SINR-K-1}
	\gamma_{2} &= \rho\Big|r^{2}_{r,2}\Big|^{2}a_{2}.
\end{align}
\vspace{-0.5cm}
\section{OP Analysis}
\label{sec:3}
 In order to reveal the benefits of the proposed downlink STAR-RIS-assisted NOMA system, we examine its performance analytically in terms of OP under different user deployments, as follows.

The OP of user $k$ is the probability that the user cannot successfully detect and remove the signals of the $k-1$ users before it and/or cannot detect its own signal. Let $\gamma^{j}_{th}$ denote the SINR threshold for user $k$ to detect user $j$'s signal, $ 1 \le j \le k $, thus, $E_{k\leftarrow j}=\left\{\gamma^{j}_{{k}}<\gamma^{j}_{th}\right\} $ denotes the event that user $k$ cannot detect user $j$'s signal. Thus, the OP for user $k$ can be obtained as
\begin{align}
	\label{OP}
	OP_{k}= 1 - P_{r}\left( E^c_{k\leftarrow 1}\cap...\cap E^c_{k\leftarrow k}\right),
\end{align}
where $E^c_{k\leftarrow j}$ is the complement of the event  $E_{k\leftarrow j}$.  

\begin{proposition} The OP of the $U_{k}$ is given by
\begin{align}
	\label{OP_11} 
	OP_{k}\left(\varrho^{*}_{k}\right)&=1 - Q_{\frac{1}{2}}\left(\frac{\mu_{k}}{\nu_{k}},\frac{\sqrt{\varrho^{*}_{k}}}{\nu_{k}}\right)\nonumber\\
	\
	& + \left(\frac{u^{2}_{k}}{\nu^{2}_{k} + u^{2}_{k}}\right)^{\frac{1}{2}}e^{\frac{\varrho^{*}_{k}}{2u^{2}_{k}}}e^{-\frac{\mu^{2}_{k}}{2\left(\nu^{2}_{k} + u^{2}_{k}\right)}}\nonumber\\
	\
	&\times Q_{\frac{1}{2}}\left(\frac{\mu_{k}}{\nu_{k}}\sqrt{\frac{u^{2}_{k}}{\nu^{2}_{k} + u^{2}_{k}}},\sqrt{\frac{\varrho^{*}_{k}\left(\nu^{2}_{k} + u^{2}_{k}\right)}{\nu^{2}_{k}u^{2}_{k}}}\right), 
\end{align}
where $\varrho^{*}_{k}= \underset{j = 1, ..., k}{\max}\left\{\varrho_{j}\right\}$, $\varrho_{j} =  \frac{\gamma^{j}_{th}}{\rho a_{j}-\rho\gamma^{j}_{th}\sum_{j}}$, $\mu_{k} = \frac{\pi}{4}\sqrt{L_{k}}N^{k}_{\chi}$, $\sum_{j}= \sum_{l= j +1 }^{K}a_{l}$, $\nu_{k} = \sqrt{\left(1-\frac{\pi^{2}}{16}\right)L_{k}N^{k}_{\chi}}$ and $u_{k} = \sqrt{0.5\rho\varrho^{*}_{k}L_{k}\left(N_{\chi}-N^{k}_{\chi}\right)}$.
\end{proposition}

\begin{proof}See Appendix \ref{sec:ApendixA}.
\end{proof}

In order to gain further insights on the OP performance, we give the following Corollary.
\begin{corollary} Considering the asymptotic behavior of OP, for $\rho\rightarrow \infty$, i.e., $\varrho^{*}_{k}\rightarrow 0$, OP is given by
\begin{align}
	\label{CDFchi}
	OP^{\infty}_{k} & = \left(\frac{\tilde{u}^{2}_{k}}{\nu^{2}_{k} + \tilde{u}^{2}_{k}}\right)^{\frac{1}{2}}e^{-\frac{\mu^{2}_{k}}{2\left(\nu^{2}_{k} + \tilde{u}^{2}_{k}\right)}},
\end{align}
\end{corollary}
\begin{proof}
		An asymptotic analysis is carried out in the high SNR region, which means that $\rho\rightarrow \infty$. Then, $\varrho^{*}_{k}= \underset{j = 1, ..., k}{\max}\left\{\frac{\gamma^{j}_{th}}{\rho a_{j}-\rho\gamma^{j}_{th}\sum_{j}}\right\}_{\rho\rightarrow \infty} = 0$. By substituting $\varrho^{*}_{k}=0$ into Proposition 1 and with some calculations besides the benefits of $Q_{\frac{1}{2}}\left(\frac{\mu_{k}}{\nu_{k}},0\right) = Q_{\frac{1}{2}}\left(\frac{\mu_{k}}{\nu_{k}}\sqrt{\frac{u^{2}_{k}}{\nu^{2}_{k} + u^{2}_{k}}},0\right) = 1$, $\tilde{u}_{k} = \sqrt{0.5\varrho^{\dagger}_{k}L_{k}\left(N_{\chi}-N^{k}_{\chi}\right)}$ and $\varrho^{\dagger}_{k} = \underset{j = 1, ..., k}{\max}\left\{\frac{\gamma^{j}_{th}}{ a_{j}-\gamma^{j}_{th}\sum_{j}}\right\}$, then Corollary 1 is obtained.
\end{proof}
\vspace{-0.2cm}
From Corollary 1, it can be noticed that the OP is not a function of SNR and it reaches a fixed value (error floor) due to subsurface-interference when the transmit SNR increases. This results in that each user will not achieve diversity gain (zero-diversity order). 

\begin{proposition} In the case of a two-user STAR-RIS-assisted NOMA network, $\left(K_{t}, K_{r}\right) = \left(1, 1\right)$, the OP of the $U_{k}$ can be expressed in closed-form as
\begin{align}
	\label{OP_11-twouser}
	OP_{k}&=1 - Q_{\frac{1}{2}}\left(\frac{\mu_{k}}{\nu_{k}},\frac{\sqrt{\varrho^{*}_{k}}}{\nu_{k}}\right).
\end{align}
Here, considering the asymptotic behavior, for $\rho\rightarrow \infty$, i.e., $\varrho^{*}_{k}\rightarrow 0$, we have $OP_{k}=0$. It is worth mentioning that in this paper, we use CLT approximation to model the statistics of SINR of the STAR-RIS-assisted NOMA networks. This approximation becomes more accurate for a large number of STAR-RIS elements\footnote{ Note that by letting $\lambda$ be arbitrarily small and/or letting $d$ be arbitrarily large as $N$ increases, the RIS far-field boundaries can be maintained, and thus, the channel model validity can be preserved \cite{path_exp}.}. However, in our proposed system, the OP becomes zero only theoretically, but in practice, it is a function of SNR as will be seen in Section V.
\end{proposition} 

\begin{proof} See Appendix \ref{sec:ApendixB}.
	\vspace{-0.2cm}
	\end{proof}
From Proposition 2, it can be noticed that in the case of a two-user STAR-RIS-assisted NOMA network, since there is no subsurface-interference, the OP becomes zero in the high SNR region. 
\vspace{-0.3cm} 
\section{Partitioning Algorithms }
\label{sec:4}
In this section, we provide the partitioning algorithms used to allocate STAR-RIS elements for different users on both sides of the surface. It is worth mentioning that the authors in [38] proposed a novel NOMA solution with RIS partitioning. In particular, they distributed the physical resources among users such that the BS and RIS were dedicated to serving different clusters of users. The RIS partitioning algorithm aims to enhance SE by improving all users' ergodic rates and maximizing user fairness. In order to achieve this goal, they formulated an optimization problem where Jain’s fairness index is the objective function to be maximized, and the proper numbers of reflecting elements needed to be assigned for users are the decision variables to be determined. Note that although the general structure of the proposed algorithms is motivated by the work of [38], they are different in the following key aspects. First, unlike the algorithms in [38], which aim to maximize user fairness [Section IV, Eq. (20a)], the proposed algorithms aim to maximize the users' sum rate while guaranteeing QoS requirements for all users. In particular, Algorithm 1 guarantees a unified OP requirement for all users, while Algorithm 2 considers the different rate requirements for individual users. Furthermore, in Algorithm 1, the OP expression is totally different from the one in the pointed-out work due to obviously different systems models. In what follows, we present the partitioning algorithms considering different users' deployments. 
\subsection{Multi-user STAR-RIS-assisted NOMA system}
In the case of a multi-user STAR-RIS-assisted NOMA system, our partitioning algorithm is carried out in two stages. First, the STAR-RIS surface is partitioned into two parts, transmission and reflection parts, which contain all the elements that operate in transmission and reflection mode, respectively. Second, each STAR-RIS part is partitioned into subsurfaces, where each subsurface is allocated to serve a specific user located on the side of that part. In the following, both stages are explained in detail.

At the first stage, we assign temporary numbers of transmitting $(\tilde{N}_t)$ and reflecting $(\tilde{N}_r)$ STAR-RIS elements for the transmission and reflection STAR-RIS parts, respectively. Furthermore, the mode of the remaining STAR-RIS elements $(N - \tilde{N}_t - \tilde{N}_r)$ will be selected further according to sum-rate and QoS requirements. $\tilde{N}_t$ and $\tilde{N}_r$ can be simply calculated, respectively as $\tilde{N}_t = K_{t}N_{thr}$ and $\tilde{N}_r = K_{r}N_{thr}$, where $N_{thr}$ is the minimum number of STAR-RIS elements that needs to be allocated to each user in order to protect it from the subsurface-interference. We use Algorithm 1 to obtain the value of $N_{thr}$ \cite{RISpartitioning}. Particularly, Algorithm 1 utilizes a probabilistic approach, where in order to protect the users from the mutual sub-surfaces interference, a common minimum OP is guaranteed for all users by finding the proper $N_{thr}$. This is achieved by keeping updating $N_{thr}$ until the OP gets below the predetermined threshold $\epsilon$. Note that the asymptotic OP in \eqref{CDFchi} is used to capture the interference effect only and ignore the different SNR levels at different users. 

On the other hand, at the second stage, we aim to determine the proper number of transmitting/reflecting STAR-RIS elements that is required for each user to maximize the sum-rate and guarantee QoS for individual users. To achieve this goal, at first, Algorithm 2 is used to determine the number of STAR-RIS elements that need to be added to $N_{thr}$ in order to fulfill the QoS requirement $R_k^{min}$ for each user. According to \eqref{rk}, the cascaded channels' gains are proportionate with the number of transmitting or reflecting elements.
\begin{algorithm}[t]
	\caption{Determination of minimum number of STAR-RIS elements that needs to be allocated to each user $(N_{thr})$.}
	\begin{algorithmic}[1]
		\REQUIRE $L_{SU,\chi,k}$, $K$, $N$, $\epsilon$.
		\STATE Initialize $N_{thr}=1$.
		\STATE \textbf{repeat}
		\STATE $\mu_{k} = \frac{\pi}{4}\sqrt{L_{k}}N_{thr}, \nu_{k} = \sqrt{\left(1-\frac{\pi^{2}}{16}\right)L_{k}N_{thr}}$,  $\tilde{u}_{k} = \sqrt{0.5\varrho^{\dagger}_{k}L_{k}\left(N-N_{thr}\right)}$.
		\STATE Obtain the asymptotic OP from \eqref{CDFchi}:\\ $OP^{\infty}_{k}	=\left(\frac{\tilde{u}^{2}_{k}}{\nu^{2}_{k} + \tilde{u}^{2}_{k}}\right)^{\frac{1}{2}}e^{-\frac{\mu^{2}_{k}}{2\left(\nu^{2}_{k} + \tilde{u}^{2}_{k}\right)}}$
		\STATE $N_{thr} = N_{thr} + 1$.
		\STATE \textbf{while} $OP^{\infty}_{k}\le\epsilon$ and $KN_{thr}\leq N$.
		\RETURN $N_{thr} = N_{thr} - 1$.
	\end{algorithmic}
\end{algorithm}
\begin{algorithm}[t]
	\caption{Determination of proper number of transmitting/reflecting STAR-RIS elements to be assigned to $U_{k}$ $(N^{k}_{\chi})$.}
	\begin{algorithmic}[1]
		\REQUIRE $N_{thr}$, $R^{min}_{k}$.
		\STATE Initialize $N^{k}_{\chi}= N^{k-1}_{\chi}$. If $k = 1$, then $N^{1}_{\chi} = N_{thr}$.
		\STATE \textbf{repeat}
		\STATE Obtain the Ergodic rate in the high SNR values:\\ $R_{k} = \text{E}\left[\log_{2}\left(1 + \gamma^{k}_{{k}}\right)\right]$. Perform the expectation over $ 10^{4} $ random channel realizations.
		\STATE $N^{k}_{\chi} = N^{k}_{\chi} + 1$. 
		\STATE \textbf{while} $R_{k}<R^{min}_{k}$.
		\RETURN $N^{k}_{\chi} = N^{k}_{\chi} - 1$.
	\end{algorithmic}
\end{algorithm}
Therefore, if the cascaded channels' gains are sorted as $\big|r^{1}_{\chi,1}\big|^{2}\le ...\le \big|r^{k}_{\chi,k}\big|^{2}\le ...\le\big|r^{k}_{\chi,K}\big|^{2} $, then the transmitting/reflecting STAR-RIS elements assigned for the users have to be also ordered as $N^{1}_{\chi}\le...\le N^{k}_{\chi} \le...\le N^{K}_{\chi}$, which is ensured by Algorithm 2. Next, in order to maximize the sum-rate, the remaining STAR-RIS elements are added to the closest user to the BS, $U_K$, which can eliminate the interference of all other users. At the end, the proper values of $N_t$ and $N_r$ are updated, respectively as $N_t = \sum_{k = 1}^{K_{t}}N^{k}_{t}$ and $N_r = \sum_{k = K_{t} + 1}^{K}N^{k}_{r}$.

In terms of the complexity, it can be readily seen that both algorithms have a finite maximum number of iterations that is upper-bounded by $N$, which guarantees their convergence. Furthermore, considering the required number of complex multiplications (CMs) as a metric, the complexity level of Algorithm 2 can be upper-bound by $\mathcal{O}(N^3)$. This can be obtained by noting that Algorithm 2 requires the calculation of $\gamma_k^k$ which, in turn, includes the vector-matrix multiplication of $r^{i}_{\chi,k}, i\neq k$, that requires $N^2$ CMs for every single iteration. On the other side, Algorithm 1 has no CM, therefore; its complexity level depends on the type of the used computational algorithm. Accordingly, the proposed Algorithms 1 and 2, have the same complexity level as the ones in \cite{RISpartitioning}.
\subsection{Two-user STAR-RIS-assisted NOMA system}
For the case of a two-user STAR-RIS-assisted NOMA system, one user is located in $\text{C}_t$ and the other user is located in $\text{C}_r$. Since,  in this case, there is no subsurface interference, we skip the first stage of the partitioning process and start directly from the second one. That is, Algorithm 2 is directly applied to fulfill the QoS requirement for $U_1$, with $N_{thr}=1$ and \eqref{SINR-11} is used to obtain $\gamma_k^k$. Next, the remaining elements need to be allocated to $U_2$, assuming $N$ is large enough to fulfill both users' QoS requirements.

\begin{remark}
As seen from Algorithms 1 and 2, the partitioning scheme does not depend on the instantaneous CSI; rather, it needs to be updated only when $K_{\chi}$, power allocation at the BS, or distances of users from the STAR-RIS change, where these parameters slowly change over several coherence channel intervals. Furthermore, after the partitioning stage and for a given channel coherence interval, BS needs to know the CSI of the BS-STAR-RIS link, which is the same for all users, and the STAR-RIS ($k$th sub-surface)-$U_k$ link, for each individual user $k$. This means that, without loss of generality, for the STAR-RIS-$U_k$ links, a maximum of $N$ pilot signals are required at the channel estimation stage. Compared to the joint phase shift design for all users (where $NK$ pilot signals are required \cite{ch-est}), an effective reduction in the channel estimation overhead is obtained by the proposed system. Furthermore, from Algorithms 1 and 2, the phase shift design problem for all users is scalable as $K_{\chi}$ changes.
\end{remark}
\section{Numerical Results}
\label{sec:5}
In this section, the performance of the considered downlink STAR-RIS-assisted multi-user wireless networks is presented via computer simulations in order to validate our theoretical analysis. We consider three cases of users' deployment, namely, $\left(K_{t}, K_{r}\right) = \left(1, 1\right)$, $\left(K_{t}, K_{r}\right) = \left(2, 1\right)$ and $\left(K_{t}, K_{r}\right) = \left(1, 2\right)$. For each case, user locations,  power allocation coefficients\footnote{Here, we focus on NOMA with fixed power allocation (F-NOMA), where a fixed amount of transmit power is allocated to each user. In particular, for the two users, the power-allocation coefficients are denoted by $a_{1}$ and $a_{2}$, where these coefficients are fixed, and $a_{1} + a_{2} = 1$ and $a_{1}\ge a_{2}$. Likewise, for the three users, $a_{1} + a_{2} + a_{3} = 1$ and $a_{1}\ge a_{2}\ge a_{1}$ \cite{F-NOMA}.}, corresponding target threshold SINR values, distances between the BS and STAR-RIS, and STAR-RIS and users are given in Table \ref{Table1} at the top of the next page. For all cases, the path gain at a reference distance of 1 meter is assumed as $\rho_0 = -30$ dB and the path loss exponents associated with the BS-STAR-RIS and the STAR-RIS-user links are given as $\left(\alpha_{BS}, \alpha_{SU}\right) = \left(2, 2\right)$ unless otherwise stated. Furthermore, we compare the proposed system with three benchmark systems, namely the classical OMA, PD-NOMA systems and uniform partitioning method by considering the same simulation parameters. Particularly, for the classical PD-NOMA system, we use the same power allocation that we use in the proposed system. On the other side, for the classical OMA system, we use the time division multiple access (TDMA) scheme, where the time slot allocation equals the power allocation in NOMA case; therefore, we always have the time slots $t_1=a_1$, $t_2=a_2$, and $t_3=a_3$. For the benchmark systems, the BS-$U_k$ path gain is obtained as $L_k=\frac{\rho_0}{d_k^\alpha}$, where $d_k$ is the BS-$U_k$ distance and $\alpha=3.5$ \cite{alpha}. In the uniform partitioning method, we divide the total number of STAR-RIS elements equally among users. This means that the number of reflecting elements equals the number of transmitting elements, which equals half of the total number of STAR-RIS elements.
\begin{figure}[t]
	\centering
	\subfloat[]{\label{fig2:a}\includegraphics[width=57mm, height=40mm]{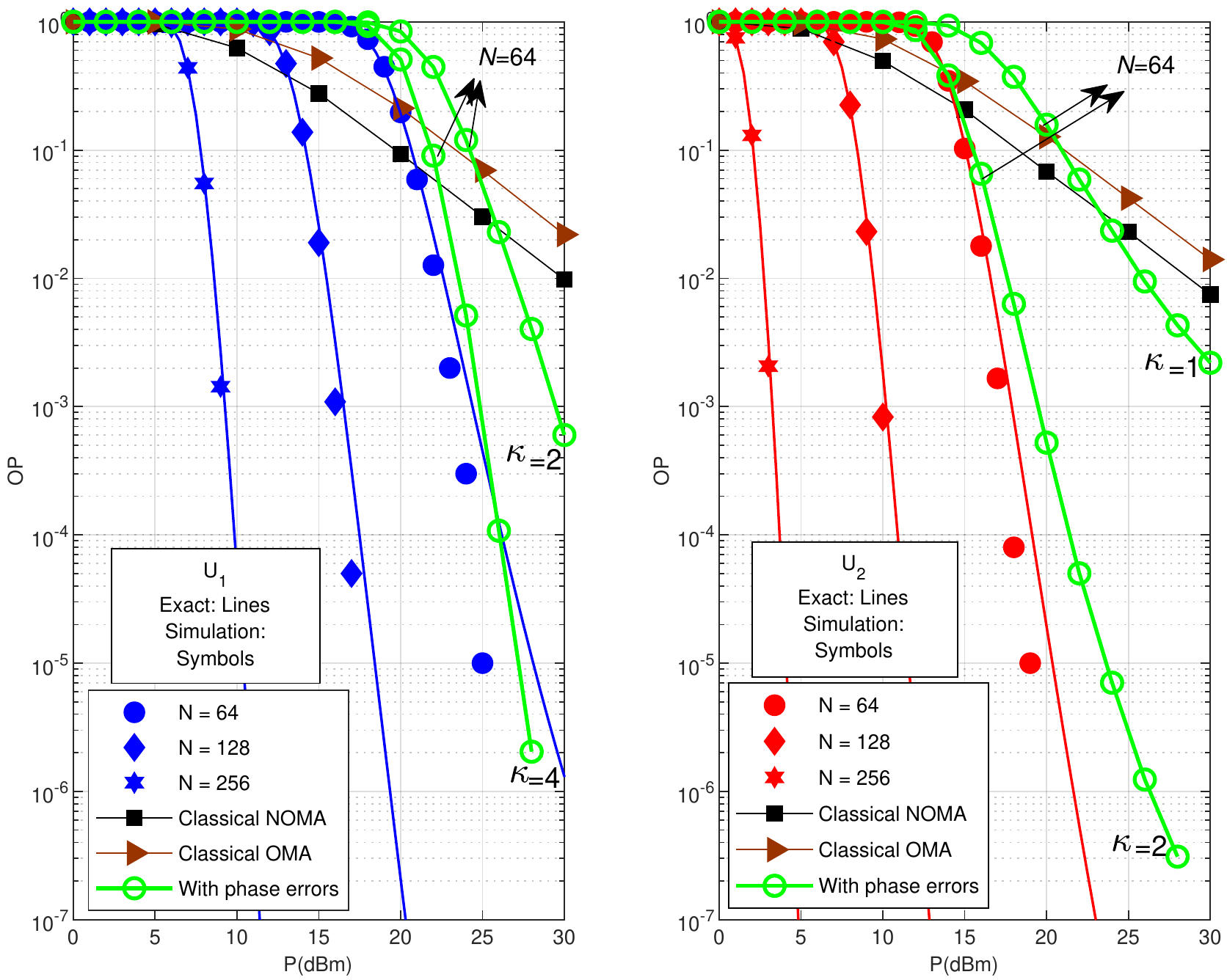}}
	\subfloat[]{\label{fig2:a}\includegraphics[width=30mm, height=40mm]{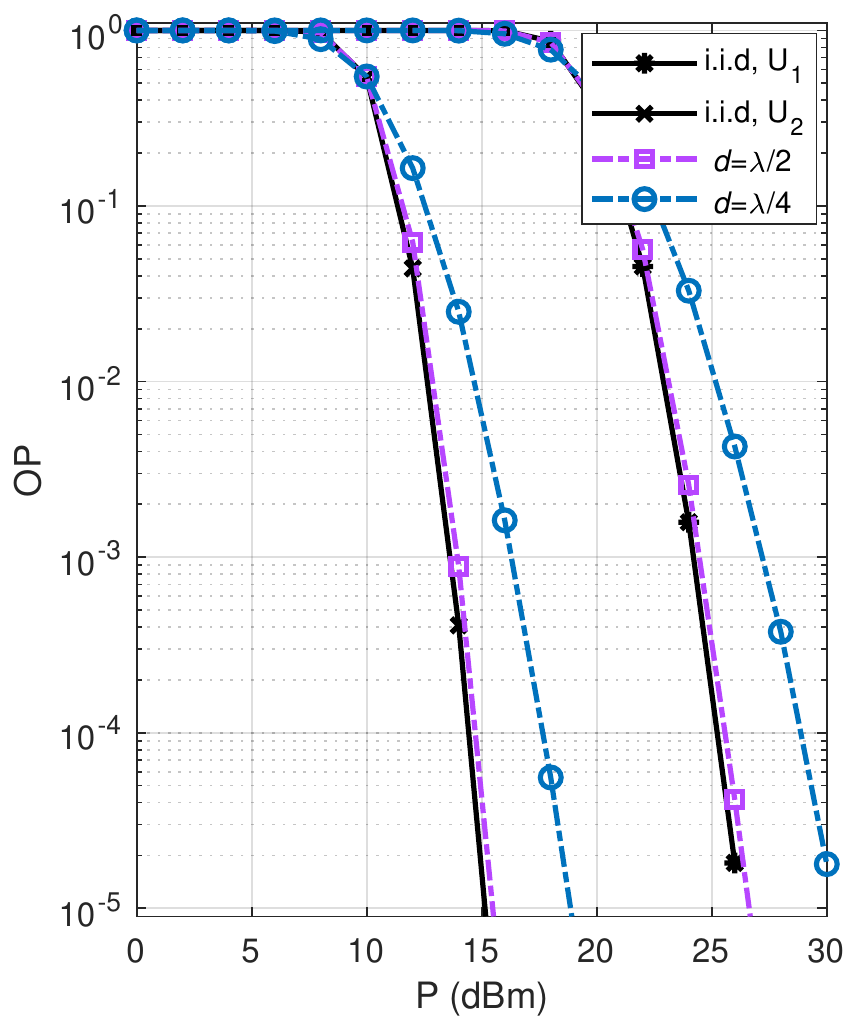}}
	\caption{OP performance in case of $\left(K_{t}, K_{r}\right) = \left(1, 1\right)$ (a) under different $N$ values and with/without phase errors, and (b) under spatial correlation effect with $N_t=24$ and $N_r=40$, STAR-RIS inter-element spacing $d$, and $1.8$ GHz operating frequency with $\lambda$ wavelength.}
	\label{Twousers}
	\vspace{-0.3cm}
\end{figure}
\begin{figure}[t!]
	\includegraphics[width=65mm,height=40mm]{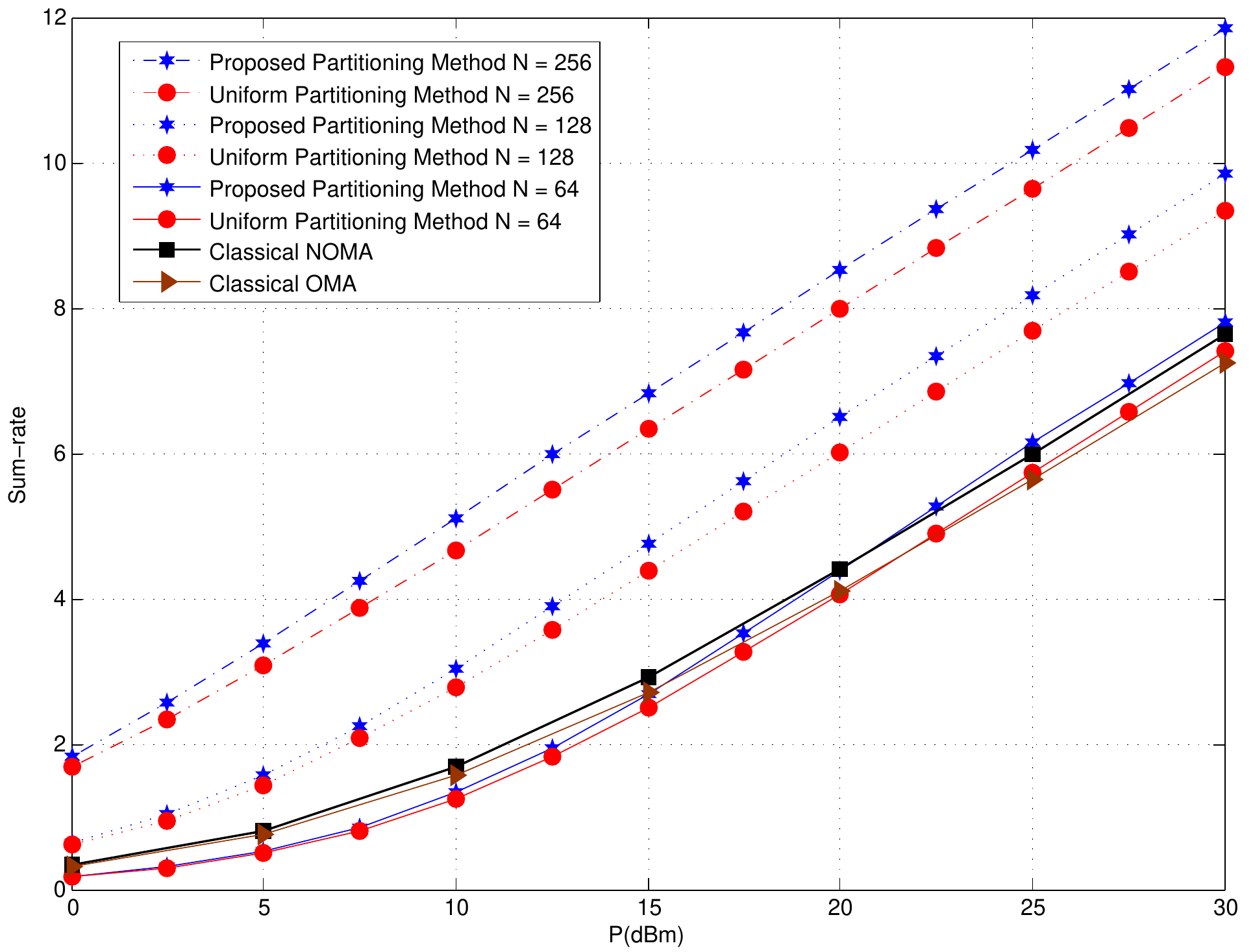}
	\centering
	\caption{The sum-rate versus $ P $ for different $N$ in the case of $\left(K_{t}, K_{r}\right) = \left(1, 1\right)$.}	
	\label{Twousers_sumrate}
	\vspace{-0.3cm}
\end{figure}
\begin{table*}[]
	\centering
	\caption{Parameters used in simulations for the three cases.}
	\vspace{-0.2cm}
	\label{Table1}
	\begin{tabular}{|
			>{\columncolor[HTML]{FFFFFF}}l |
			>{\columncolor[HTML]{FFFFFF}}l |
			>{\columncolor[HTML]{FFFFFF}}l |}
		\hline 
		\multicolumn{2}{|l|}{\cellcolor[HTML]{FFFFFF}User locations}               & \begin{tabular}[c]{@{}l@{}}Case 1: $U_{1}$ is located in $\text{C}_t$ and $U_{2}$ is located in $\text{C}_r$.\\ Case 2: $U_{1}$ and $ U_{2}$ are located in $\text{C}_t$, and $U_{3}$ is located in $\text{C}_r$.\\ Case 3: $U_{1}$ is located in $\text{C}_t$, and $U_{2}$ and $ U_{3}$ are located in $\text{C}_r$.\end{tabular}                                                                \\ \hline
		\multicolumn{2}{|l|}{\cellcolor[HTML]{FFFFFF}Power allocation coefficients} & \begin{tabular}[c]{@{}l@{}}Case 1: $\left(a_{1}, a_{2}\right)=\left(0.6, 0.4\right)$\\ Case 2: $\left(a_{1}, a_{2}, a_{3}\right)=\left(0.6, 0.3, 0.1\right)$\\ Case 3: $\left(a_{1}, a_{2}, a_{3}\right)=\left(0.6, 0.3, 0.1\right)$\end{tabular}                                                                                  \\ \hline
		\multicolumn{2}{|l|}{\cellcolor[HTML]{FFFFFF}Target threshold SINR values}  & \begin{tabular}[c]{@{}l@{}}Case 1: $\left(\gamma^{1}_{th}, \gamma^{2}_{th}\right) = \left(1,1\right)$\\ Case 2: $\left(\gamma^{1}_{th}, \gamma^{2}_{th}, \gamma^{3}_{th}\right) = \left(0.7, 0.7, 0.7\right)$\\ Case 3: $\left(\gamma^{1}_{th}, \gamma^{2}_{th}, \gamma^{3}_{th}\right) = \left(0.5, 0.5, 0.3\right)$\end{tabular} \\ \hline
		\multicolumn{2}{|l|}{\cellcolor[HTML]{FFFFFF}Distances (in meters)}                     & \begin{tabular}[c]{@{}l@{}}Case 1: $\left(d_{BS}, d_{SU,t,1}, d_{SU,r,2}\right) = \left(50, 50, 40 \right)$\\ Case 2: $\left(d_{BS}, d_{SU,t,1}, d_{SU,t,2}, d_{SU,r,3}\right) = \left(20, 50, 40, 30 \right)$\\ Case 3: $\left(d_{BS}, d_{SU,t,1}, d_{SU,r,2}, d_{SU,r,3}\right) = \left(20, 50, 40, 30 \right)$\end{tabular}
		 \\ \hline
	\end{tabular}
\end{table*}
\subsection{Case 1: $\left(K_{t}, K_{r}\right) = \left(1, 1\right)$} 

As shown in Fig. \ref{Twousers}(a), our theoretical result in \eqref{OP_11-twouser} using the CLT is considerably accurate for increasing $ N $ values and the OP performance enhances as number of STAR-RIS elements increases. For instance, to achieve $OP_{1}$ = $10^{-3}$, 10 dBm and 5 dBm gain advantages are obtained in $ P $ for $N = 128$ over $N = 64$ and $N = 256$ over $128$, respectively. On the other hand, to achieve $OP_{2}$ = $10^{-3}$, about 7.5 dBm and 7 dBm gain advantages are observed in $ P $ for $N = 128$ over $N = 64$ and $N = 256$ over $128$, respectively. Furthermore, the effect of phase errors is shown to significantly reduce the OP performance as the concentration parameter $\kappa$ decreases (flatter error distribution)\footnote{We consider Von Mises distribution to simulate the phase errors resulting from phase estimation errors and/or phase quantization in the case of discrete phase shifts \cite{phs-errs}.}, yet, the performance of the proposed system is still superior to the benchmark ones at the high $P$ region. In Fig. \ref{Twousers}(b), the spatial correlation effect is shown for two different inter-element spacing values $d$, where the spatial correlation model proposed for the RIS in \cite{spat-corr} is used. Note that, for $d=\lambda/2$, the curve obtained under the spatial correlation assumption is close to the one under the independent and identically distributed (i.i.d.) assumption of the RIS channels, which makes the theoretical OP obtained in Section III (under the i.i.d. assumption) applicable to the spatial correlation case. Furthermore, the spatial correlation (as $d$ decreases) is shown to have a negative impact on the OP performance due to the lack of diversity in the RIS channels.

Note that in Figs. \ref{Twousers}-\ref{Twousers_sumrate}, according to Algorithm 2, the number of the transmitting STAR-RIS elements for $U_{1}$ and the number of the reflecting STAR-RIS elements for $U_{2}$ are determined, respectively as $N^{1}_{t} = \left\{26,51,102\right\}$ and $N^{2}_{r} = \left\{38,77,154\right\}$. We can clearly notice from Fig. \ref{Twousers_sumrate} that the sum-rate proportionally increases as $ P $ increases. This is because the sum-rate depends remarkably on the strongest user ($U_{2}$) and in this case, $U_{2}$ does not suffer from any subsurface-interference. In addition, the proposed partitioning method outperforms the uniform partitioning method. Also, from both figures, we can see the superiority of our proposed system over the classical NOMA and OMA particularly by increasing the total number of the elements of STAR-RIS. For example, for $U_{1}$, to achieve OP$_{1}$ = $10^{-2}$, 10 dB, 17 dB and 23 dB SNR gain advantages are obtained for proposed system with $N =$ 64, 128 and 256 over classical NOMA, respectively. In addition, for $U_{1}$, to achieve OP$_{1}$ = $2\times10^{-2}$, 8 dB, 15 dB and 23 dB SNR gain advantages are obtained for proposed system with $N =$ 64, 128 and 256 over classical OMA, respectively.
\subsection{Case 2: $\left(K_{t}, K_{r}\right) = \left(2, 1\right)$}
Figs. \ref{Case21} and \ref{Case21sumrate} demonstrate the performance of the STAR-RIS-assisted NOMA network in terms of OP and sum-rate, respectively for different number of STAR-RIS elements. The number of transmitting/reflecting STAR-RIS elements assigned for each user is provided in Table \ref{Table2}. It can be seen from Fig. \ref{Case21} that when $ P $ increases, the $OP_{1}$ and $OP_{2}$ reach an error floor due to subsurface-interference. On the other hand, $OP_{3}$ does not reach error floor since $U_{3}$ does not suffer from any subsurface-interference. Compared to the classical NOMA and OMA, the proposed system is always better for $U_{3}$ in the whole $ P $ range as well as $U_{1}$ and $U_{2}$ except in the high $ P $ regime. As shown in Fig. \ref{Case21sumrate}, the sum-rate still proportionally increases as $ P $ increases since the strongest user $(U_{3})$ does not suffer from the subsurface-interference. Furthermore, the STAR-RIS requires about $ N = 90 $ and $N = 150$ elements to outperform classical OMA and classical NOMA, respectively.
\begin{table}[t]
	\centering
	\caption{Number of STAR-RIS elements assigned for each user in Case 2 when $N \in \left\{60, 90, 120, 150\right\} $.}
	\label{Table2}
	\begin{tabular}{|c|c|c|c|c|c|}
		\hline
		$N$ & $N^1_t$ & $N^2_t$ & $N^3_r$ & $N_t$ & $N_r$  \\ \hline
		60  & 16        & 20    & 24   &  36        & 24                 \\ \hline
		90  & 24        & 30    & 36   & 54         & 36                 \\ \hline
		120 & 32        & 40   & 48   & 72         & 48                 \\ \hline
		150 & 36        & 50   & 64   & 86         & 64               \\ \hline
	\end{tabular}
\end{table}
\begin{figure}[t!]
	\includegraphics[width=90mm,height=43mm]{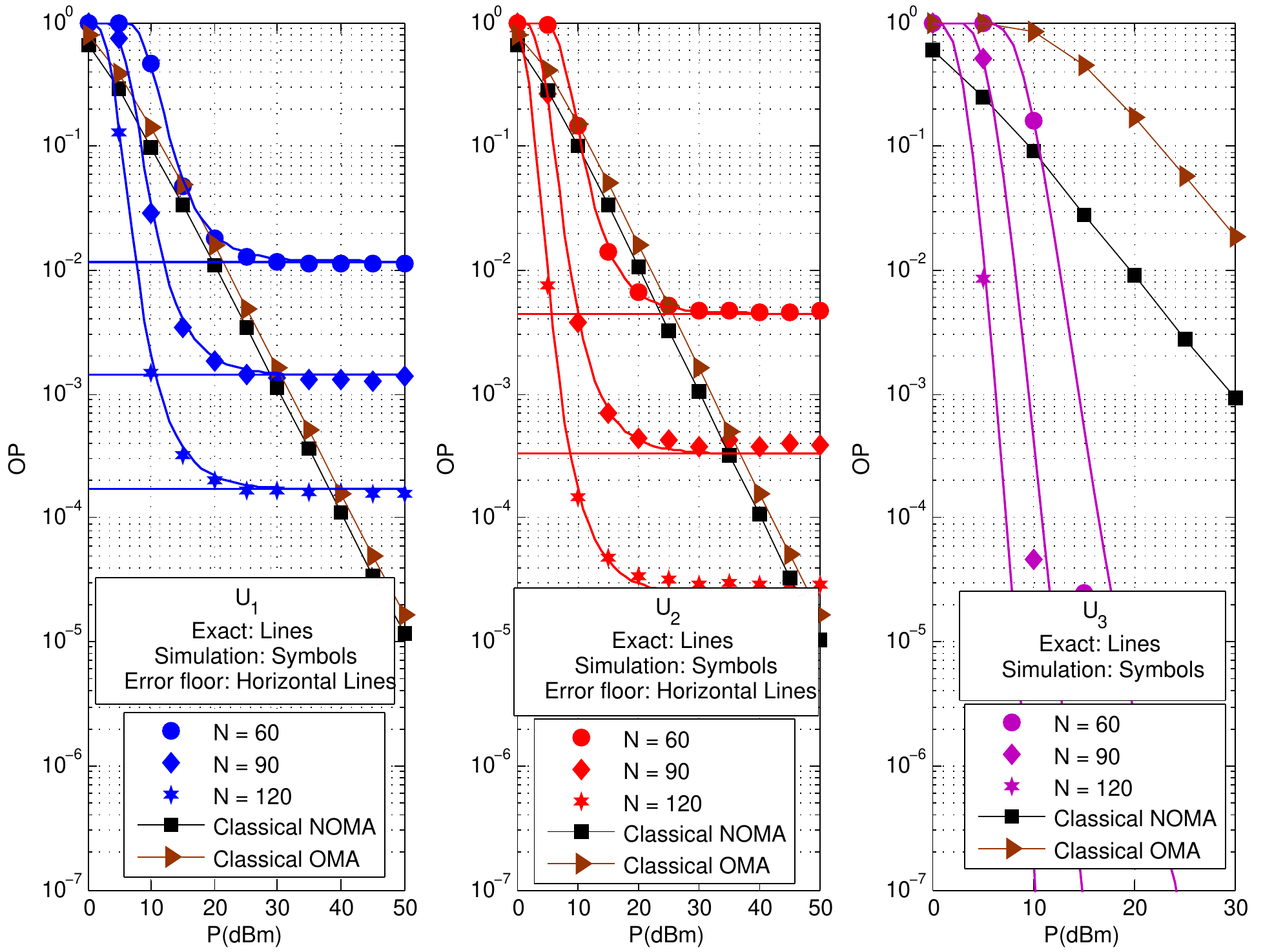}
	\centering
	\caption{The OP versus $ P $ for different $N$ in the case of $\left(K_{t}, K_{r}\right) = \left(2, 1\right)$.}	
	\label{Case21}
		\vspace{-0.5cm}
\end{figure}
\begin{figure}[t!]
	\includegraphics[width=65mm,height=40mm]{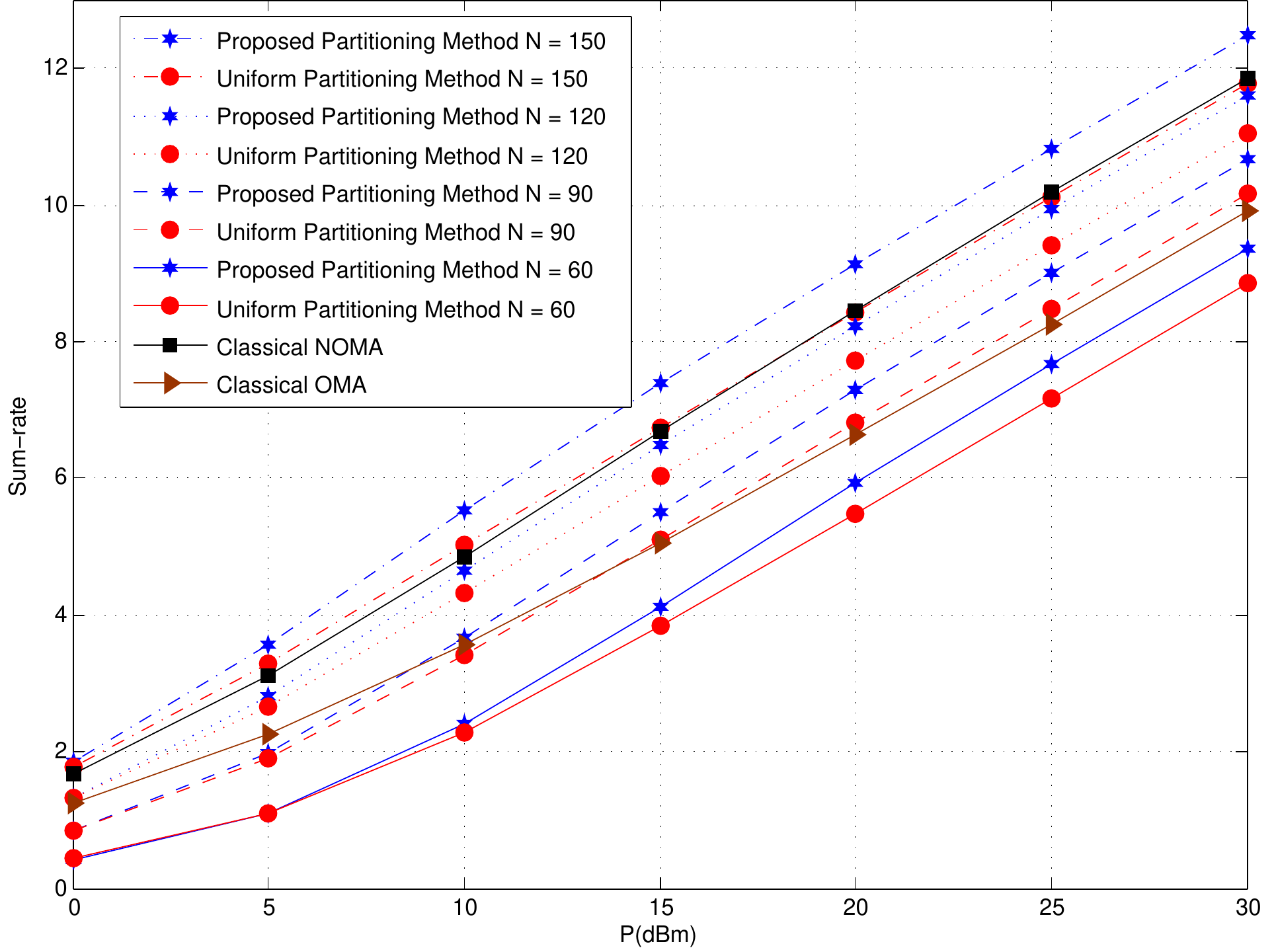}
	\centering
	\caption{The sum-rate versus $ P $ for different $N$ in the case of $\left(K_{t}, K_{r}\right) = \left(2, 1\right)$.}	
	\label{Case21sumrate}
		\vspace{-0.5cm}
\end{figure}
	\vspace{-0.5cm}
\subsection{Case 3: $\left(K_{t}, K_{r}\right) = \left(1, 2\right)$}
Figs. \ref{Case12} and \ref{Case12sumrate} present the OP and sum-rate performance of the STAR-RIS-assisted NOMA network for different number of STAR-RIS elements. The number of transmitting/reflecting STAR-RIS elements assigned for each user is provided in Table \ref{Table3}. From Fig. \ref{Case12}, it is clearly observed that the $OP_{2}$ and $OP_{3}$ reach an error floor due to subsurface-interference. On the other hand, the $OP_1$ does not reach error floor since $U_{1}$ does not suffer from any subsurface-interference. As seen from Fig. \ref{Case12sumrate}, the sum-rate saturates in the high $ P $ region. This is because the strongest user $(U_{3})$ suffers from subsurface-interference, which significantly affect the sum-rate performance. The proposed partitioning approach is also superior to the uniform partitioning approach. Although classical NOMA and OMA outperform the proposed system in the high SNR region, with a large number of STAR-RIS elements, the proposed system can outperform both classical systems. For instance, the proposed system requires more than $N = 390$ elements to outperform both classical systems.
\begin{table}[t!]
	\centering
	\caption{Number of STAR-RIS elements assigned for each user in Case 3 when $N \in  \left\{60, 90, 120, 180\right\} $.}
	\vspace{-0.2cm}
	\label{Table3}
	\begin{tabular}{|c|c|c|c|c|c|}
		\hline
		$N$ & $N^1_t$ & $N^2_r$ & $N^3_r$ & $N_t$ & $N_r$ \\ \hline
		60  & 16        & 20   & 24   & 16         & 44                 \\ \hline
		90  & 19        & 30   & 41   & 19         & 71                 \\ \hline
		120 & 22        & 40   & 58   & 22         & 98                \\ \hline
		180 & 35        & 60   & 85  & 35          & 145                \\ \hline
		390 & 52        & 130   & 208  & 52          & 338                \\ \hline
	\end{tabular}
\end{table}
\section{Conclusion}
\label{sec:6}
\begin{figure}[t!]
	\includegraphics[width=90mm,height=43mm]{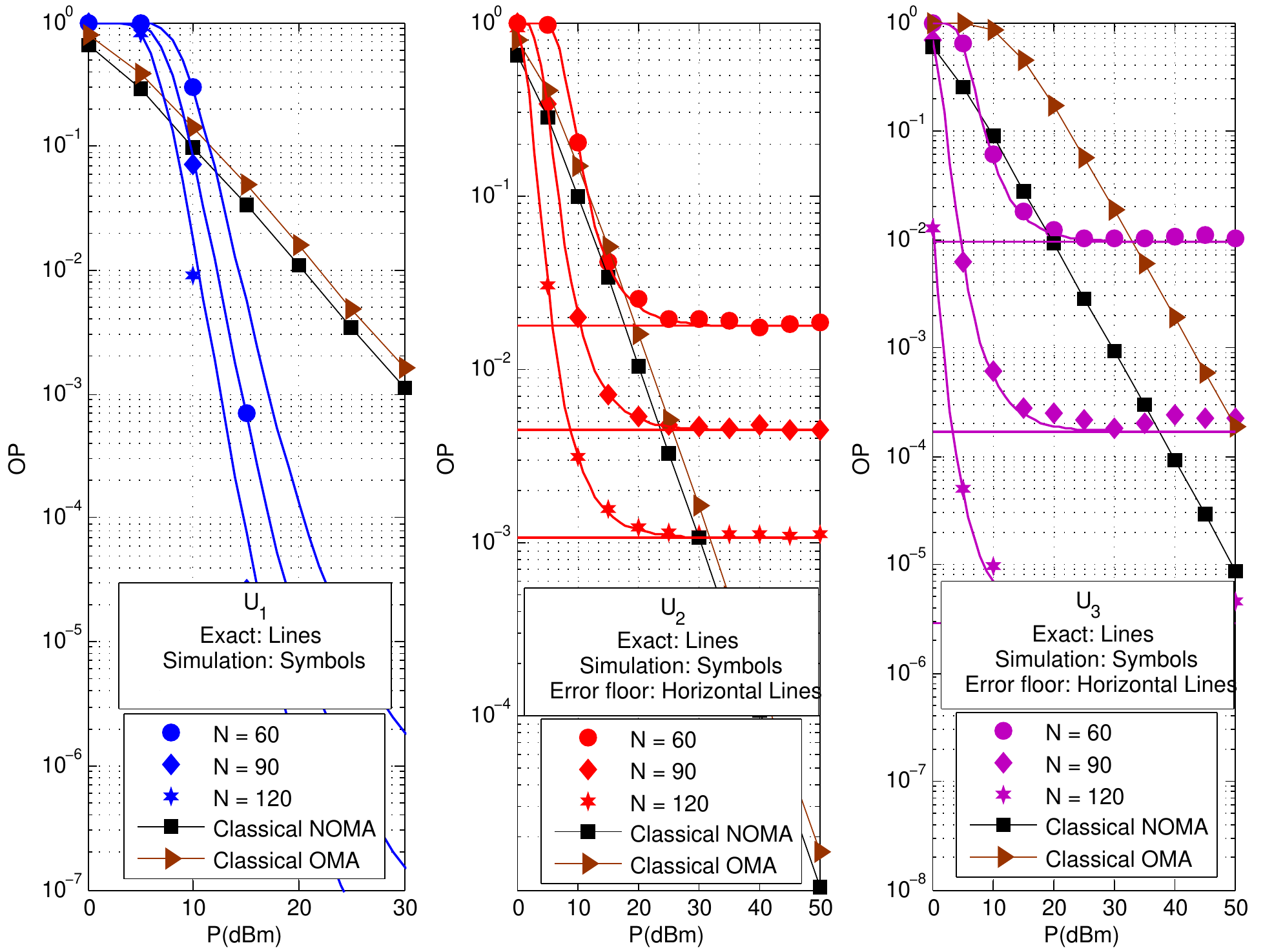}
	\centering
	\caption{The OP versus $ P $ for different $N$ in the case of $\left(K_{t}, K_{r}\right) = \left(1, 2\right)$.}	
	\label{Case12}
	\vspace{-0.4cm}
\end{figure}
\begin{figure}[t!]
	\includegraphics[width=65mm,height=40mm]{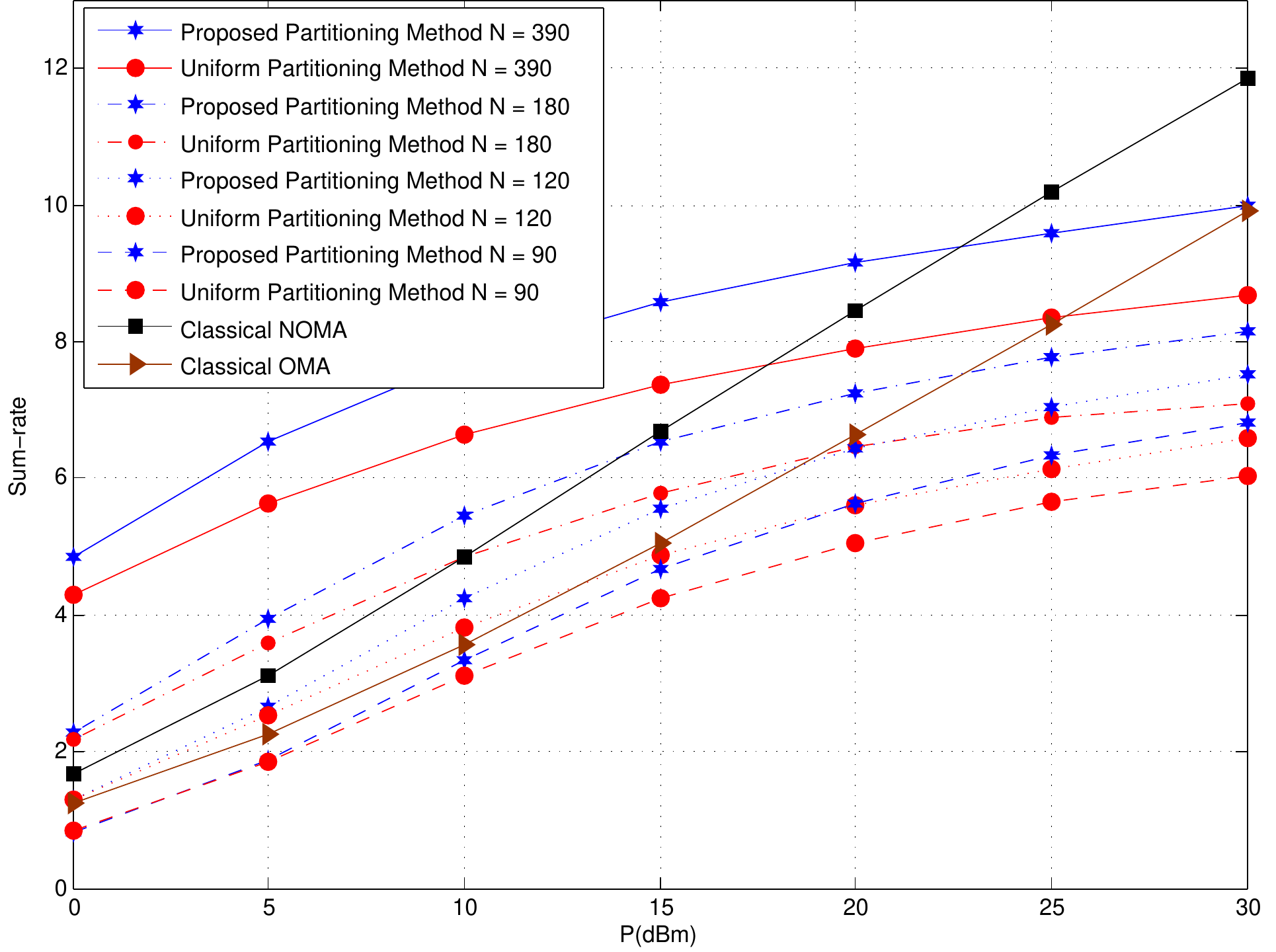}
	\centering
	\caption{The sum-rate versus $ P $ for different $N$ in the case of $\left(K_{t}, K_{r}\right) = \left(1, 2\right)$.}	
	\label{Case12sumrate}
	\vspace{-0.7cm}
\end{figure}
In this work, we have proposed a STAR-RIS-assisted NOMA network where the MS operating protocol is utilized to extend the coverage area of classical RIS. Next, we have designed a novel STAR-RIS partitioning algorithm to serve NOMA users on both sides of the STAR-RIS. The proposed algorithm aims to determine the proper number of transmitting/reflecting STAR-RIS elements that requires to be allocated to each user to maximize the sum-rate while fulfilling the QoS requirement. For different NOMA user deployments, we have provided a closed-form expression of the exact OP. Then, a asymptotic analysis has been conducted in order to study the system performance at the high SNR. We note that the OP reaches an error floor in the high SNR region due to the impact of the STAR-RIS's subsurface-interference and it can be reduced by increasing the number of STAR-RIS elements. Finally, simulation results have validated that the proposed partitioning algorithm outperforms the uniform partitioning algorithm. Moreover, the proposed system can achieve better performance in terms of OP and sum-rate than the classical NOMA and OMA. We conclude that the use of STAR-RIS NOMA has the potential to provide an enhanced version of classical NOMA systems, where unlike the classical RIS, the STAR-RIS doubled the coverage area. Furthermore, in addition to the STAR-RIS MS mode, our proposed partitioning scheme has added more degrees of freedom to the system design, compared to the classical PD-NOMA protocol, by flexibly and dynamically allocating the physical resources (STAR-RIS elements) among users. In the future, we extend the proposed system to STAR-RIS aided NOMA networks with multiple antennas.
\appendices
\vspace{-0.5cm}
\section{Proof of Proposition 1}
\label{sec:ApendixA}
By considering \eqref{SINR},  $E^c_{k\leftarrow j}$ can be expressed as
\vspace{-0.1cm}
{\small
	\begin{align}
\label{SINRevent}
E^c_{k\leftarrow j} &=\nonumber\\
\
& \left\{\frac{\rho\Big|r^{k}_{\chi,k}\Big|^{2}a_{j}}{\rho\Big|r^{k}_{\chi,k}\Big|^{2}\sum_{j} + \rho\Big|\sum_{\underset{i\ne k}{i \in\mathbb{K}_{\chi}}}^{}r^{i}_{\chi,k}\Big|^{2} + 1}>\gamma^{j}_{th}\right\}\nonumber\\
\
&=^{\text{(c)}}\Bigg\{\Big|r^{k}_{\chi,k}\Big|^{2}> \rho\varrho_{j}\Big|\sum_{\underset{i\ne k}{i \in\mathbb{K}_{\chi}}}^{}r^{i}_{\chi,k}\Big|^{2} + \varrho_{j}\Bigg\},
\end{align}}
\vspace{-0.1cm}
where $\sum_{j}= \sum_{l= j +1 }^{K}a_{l}$, $\varrho_{j} =  \frac{\gamma^{j}_{th}}{\rho a_{j}-\rho\gamma^{j}_{th}\sum_{j}}$ and the step (c) is obtained under the condition of $a_{j}>\gamma^{j}_{th}\sum_{l= j +1 }^{K}a_{l}$. Now, substituting \eqref{SINRevent} into \eqref{OP} and defining $\varrho^{*}_{j}= \underset{j = 1, ..., k}{\max}\left\{\varrho_{j}\right\}$, then the OP of the $U_{k}$ can be stated as
{\small
	\begin{align}
\label{OP_1} 
OP_{k} &= 1 - P_{r}\Bigg(\Big|r^{k}_{\chi,k}\Big|^{2}> \rho\varrho^{*}_{k}\Big|\sum_{\underset{i\ne k}{i \in\mathbb{K}_{\chi}}}^{}r^{i}_{\chi,k}\Big|^{2} +\varrho^{*}_{k} \Bigg)\nonumber\\
\
& = P_{r}\bigg(\Big|r^{k}_{\chi,k}\Big|^{2}< \rho\varrho^{*}_{k}\Big|\sum_{\underset{i\ne k}{i \in\mathbb{K}_{\chi}}}^{}r^{i}_{\chi,k}\Big|^{2} +\varrho^{*}_{k} \bigg)
\nonumber\\
\
&=F_{{R}_{k}}\left(\varrho^{*}_{k}\right).
\end{align}}
Here, ${R}_{k} = \big|r^{k}_{\chi,k}\big|^{2}-\rho\varrho^{*}_{k}\big|\sum_{\underset{i\ne k}{i \in\mathbb{K}_{\chi}}}^{}r^{i}_{\chi,k}\big|^{2}= \big|r^{k}_{\chi,k}\big|^{2}-\big|\sqrt{\rho\varrho^{*}_{k}}\sum_{\underset{i\ne k}{i \in\mathbb{K}_{\chi}}}^{}r^{i}_{\chi,k}\big|^{2}$. It is worth noting that the BS-STAR-RIS-$U_k$ cascaded channel gains are sorted as $\big|r^{1}_{\chi,1}\big|^{2}\le ...\le \big|r^{k}_{\chi,k}\big|^{2}\le ...\le\big|r^{k}_{\chi,K}\big|^{2} $, where this users' order can be always guaranteed by allocating the proper START-RIS elements for each user. 

Noting $ \zeta^{i,n}_{\chi} $ and $ \eta^{i,n}_{\chi,k} $ are independently Rayleigh distributed random variables (RVs) with a mean $\mathrm{E}\left[\zeta^{i,n}_{\chi}\right]=\mathrm{E}\left[\eta^{i,n}_{\chi,k}\right]=\frac{\sqrt{\pi}}{2}$ and a variance $\mathrm{VAR}\left[\zeta^{i,n}_{\chi}\right]=\mathrm{VAR}\left[\eta^{i,n}_{\chi,k}\right]=1-\frac{\pi}{4}$. Then, $\mathrm{E}\left[\zeta^{i,n}_{\chi}\eta^{i,n}_{\chi,k}\right]=\frac{\pi}{4}$ and $\mathrm{VAR}\left[\zeta^{i,n}_{\chi}\eta^{i,n}_{\chi,k}\right]=1-\frac{\pi^{2}}{16}$. According to the CLT, for $N^{k}_{\chi}>>1$, the cascaded channel ${{r}}^{k}_{\chi,k}$ converges to Gaussian RV. Thus, 
\begin{align}
	\label{rkk}
	{{r}}^{k}_{\chi,k} \sim \mathcal{N}\left(\frac{\pi}{4}\sqrt{L_{k}}N^{k}_{\chi},\left(1-\frac{\pi^{2}}{16}\right)L_{k}N^{k}_{\chi}\right).
\end{align}
In addition, for $N_{\chi}-N^{k}_{\chi}>>1$, $\sum_{i\ne k}^{K_{\chi}}r^{i}_{\chi,k}$ converges to Gaussian RV. So, 
\begin{align}
	\label{rii}
	\sum_{\underset{i\ne k}{i \in\mathbb{K}_{\chi}}}^{}r^{i}_{\chi,k}\sim \mathcal{CN}\left(0,L_{k}\left(N_{\chi}-N^{k}_{\chi}\right)\right).
\end{align} 
It is noticed that $\big|{{r}}^{k}_{\chi,k}\big|^{2}$ is a non-central chi-square RV with one degree of freedom and $\big|\sum_{\underset{i\ne k}{i \in\mathbb{K}_{\chi}}}^{}r^{i}_{\chi,k}\big|^{2}$ is a central chi-square RV with two degrees of freedom. Hence, $ {R}_{k} $ is the difference of a non-central and central independent chi-square RVs and by using its corresponding CDF $F_{{R}_{k}}\left(x\right)$ \cite{Chi}, the OP of the $U_{k}$ can be expressed in closed-form as in \eqref{OP_11}.
\section{Proof of Proposition 2}
\label{sec:ApendixB}	
By using \eqref{SINR-1}, the OP of the $U_{k}$ can be stated as
{\small
	\begin{align}
	\label{OP_1s} 
	OP_{k} &=P_{r}\left(\frac{\rho\Big|r^{k}_{\chi,k}\Big|^{2}a_{j}}{\rho\Big|r^{k}_{\chi,k}\Big|^{2}\sum_{j} +  1}<\gamma^{j}_{th} \right)\nonumber\\
	\
	&=P_{r}\left(\Big|r^{k}_{\chi,k}\Big|^{2}\left(\rho a_{j} - \rho\gamma^{j}_{th}\sum_{j}\right)<\gamma^{j}_{th} \right)\nonumber\\
	\
	&=P_{r}\left(\Big|r^{k}_{\chi,k}\Big|^{2}< \varrho^{*}_{k} \right)
	\nonumber\\
	\
	&=F_{\big|r^{k}_{\chi,k}\big|^{2}}\left(\varrho^{*}_{k}\right).
\end{align}}
Recall, $\big|\mathbf{{r}}^{k}_{\chi,k}\big|^{2}$ is a non-central chi-square RV with one degree of freedom and its corresponding CDF can be written as \cite{Chi}
\begin{align}
	\label{CDFrkunsorted}
	F_{\big|\mathbf{{r}}^{k}_{\chi,k}\big|^{2}}\left(x\right)&=1 - Q_{\frac{1}{2}}\left(\frac{\mu_{k}}{\nu_{k}},\frac{\sqrt{x}}{\nu_{k}}\right).
\end{align}
Then, by using \eqref{CDFrkunsorted}, the OP of the $U_{k}$ can be stated as in \eqref{OP_11-twouser}.
\end{document}